\documentclass[preliminary,copyright,creativecommons]{eptcs}
\providecommand{\event}{AFL 2017} 
\usepackage{underscore}

\usepackage{amsthm}
\usepackage{amssymb}
\usepackage{amsmath}
\usepackage{graphicx}
\usepackage{enumitem}

\usepackage{url}

\usepackage{pgf,tikz}
\usetikzlibrary{arrows,automata,positioning,calc,shapes}
\usetikzlibrary{decorations.pathmorphing}

\newcommand*{\dfa}{\textsc{dfa}}
\newcommand*{\nfa}{\textsc{nfa}}
\newcommand*{\dfas}{\textsc{dfa}s}
\newcommand*{\revdfa}{\textsc{rev-dfa}}
\newcommand*{\revdfas}{\textsc{rev-dfa}s}
\newcommand*{\krevdfa}[1]{\textsc{rev$_{\!#1}$-dfa}}
\newcommand*{\krevdfas}[1]{\textsc{rev$_{\!#1}$-dfa}s}
\newcommand*{\rev}{\textsc{rev}}
\newcommand*{\krev}[1]{\textsc{rev$_{\!#1}$}}

\renewcommand*{\sigma}{a}

\newtheorem{theorem}{Theorem}
\newtheorem{lemma}{Lemma}
\newtheorem{corollary}{Corollary}

\newtheorem{definition}{Definition}

\newtheorem{ex}{Example}
\newenvironment{example}{\begin{ex}\rm}{\end{ex}}
\newenvironment{example-cont}[1]{\bigskip\noindent\textbf{Example~\ref{#1}.~(cont.)\hspace{\labelsep}}}{\bigskip\noindent}

\newcounter{letter}

\renewcommand{\cal}{\mathcal}

\newcommand{\scc}{\textsc{scc}}
\newcommand{\sccs}{\textsc{scc}s}

\title{Weakly and Strongly Irreversible Regular Languages}
\author{Giovanna J.~Lavado \and Giovanni Pighizzini \and Luca Prigioniero
	\institute{
		Dipartimento di Informatica, Universit\`{a} degli Studi di Milano\\
		via\ Comelico 39/41, 20135 Milano, Italy
	}
	\email{$\{$lavado$,$ pighizzini$,$ prigioniero$\}$@di.unimi.it}
}
\def\titlerunning{Weakly and Strongly Irreversible Regular Languages}
\def\authorrunning{G.J. Lavado, G. Pighizzini \& L. Prigioniero}

\begin{document}
	\maketitle
	\begin{abstract}
		\noindent
		Finite automata whose computations can be reversed, at any point, by
		knowing the last~$k$ symbols read from the input, for a fixed~$k$, are
		considered.
		These devices and their accepted languages are called~$k$-reversible
		automata and~$k$-reversible languages, respectively.
		The existence of~$k$-reversible languages which are
		not~$(k-1)$-reversible is known, for each~$k>1$.
		This gives an infinite hierarchy of \emph{weakly irreversible
		languages}, i.e., languages which are~$k$-reversible for some~$k$.
		Conditions characterizing the class of~$k$-reversible languages, for
		each fixed~$k$, and the class of weakly irreversible languages are
		obtained.
		From these conditions, a procedure that given a finite automaton
		decides if the accepted language is weakly or strongly (i.e., not
		weakly) irreversible is described.
		Furthermore, a construction which allows to transform any finite
		automaton which is not~$k$-reversible, but which accepts
		a~$k$-reversible language, into an equivalent~$k$-reversible finite
		automaton, is presented.
	\end{abstract}

	\section{Introduction}
	The principle of reversibility, which is fundamental in thermodynamics, has
	been widely investigated for computational devices.
	The first works on this topic already appeared half a century ago and are
	due to Landauer and Bennet~\cite{Lan61,Ben73}.
	More recently, several papers presenting investigations on reversibility in
	space bounded Turing machines, finite automata, and other devices appeared
	in the literature~(see, e.g.,~\cite{Ang82, Pin92, KondacsW97, LMT00, Lom02,
	HJK15, LPP16}).

	A process is said to be reversible if its reversal causes no changes in the
	original state of the system.
	In a similar way, a computational device is said to be reversible when each
	configuration has at most one predecessor and one successor, thus implying
	that there is no loss of information during the computation.
	As observed by Landauer, logical irreversibility is associated with
	physical irreversibility and implies a certain amount of heat generation.
	Hence, in order to avoid power dissipation and to reduce the overall power
	consumption of computational devices, it can be interesting to realize
	reversible devices.

	In this paper we focus on finite automata.
	While each two-way finite automaton can be converted into an equivalent one
	which is reversible~\cite{KondacsW97}, in the case of one-way finite
	automata (that, from now on, will be simply called \emph{finite automata})
	this is not always possible, namely there are regular languages as, for
	instance, the language~$a^*b^*$, that are recognized only by finite automata
	that are not reversible~\cite{Pin92}.

	In~\cite{HJK15}, the authors gave an automata characterization of the class
	of \emph{reversible languages}, i.e., the class of regular languages which
	are accepted by reversible automata: a language is reversible if and only
	if the minimum deterministic automaton accepting it does not contain a
	certain \emph{forbidden pattern}.
	Furthermore, they provide a construction to transform a deterministic
	automaton not containing such forbidden pattern into an equivalent
	reversible automaton.
	This construction is based on the replication of some strongly connected
	components in the transition graph of the minimum automaton.
	Unfortunately, this can lead to an exponential increase in the number of
	the states, which, in the worst case, cannot be avoided.
	To overcome this problem, two techniques for representing reversible
	automata, without explicitly describing replicated parts, have been
	obtained in~\cite{LP17}.

	In this paper, we deepen these investigations, by introducing the notions
	of \emph{weakly} and \emph{strongly irreversible} language.
	By definition, a reversible automaton  during a computation is able to move
	back from a configuration (state and input head position) to the previous
	one by knowing the last symbol which has been read from the input tape.
	This is equivalent to saying that all transitions entering in a same state are
	on different input symbols.
	Now, suppose to give the possibility to automata to see back more than one
	symbol on the input tape, in order to move from a configuration to the
	previous one. Does this possibility enlarge the class of languages accepted
	by reversible (in this extended sense) automata? It is not difficult to
	give a positive answer to this question.

	Considering this idea, we recall the notion of~$k$-reversibility: a regular
	language is~\emph{$k$-reversible} if it is accepted by a finite automaton
	whose computations can be reversed by knowing the sequence of the last~$k$
	symbols that have been read from the input tape.
	This notion was previously introduced in~\cite{KW14} by proving the
	existence of an infinite hierarchy of \emph{degrees of irreversibility}:
	for each~$k>1$ there exists a language which is~$k$-reversible
	but not~$(k-1)$-reversible.
	Here we prove that there are regular languages which are not~$k$-reversible
	for any~$k$.
	Such languages are called~\emph{strongly irreversible}, in contrast with
	the other regular languages which are called~\emph{weakly irreversible}.

	As in the case of ``standard'' reversibility (or $1$-reversibility), we
	provide an automata characterization of the classes of weakly and strongly
	irreversible languages. Indeed, generalizing the notion of forbidden
	pattern presented in~\cite{HJK15}, we show that a language
	is~$k$-reversible if and only if the minimum automaton accepting it does
	not contain a certain~\emph{$k$-forbidden pattern}.
	We also give a construction to transform each automaton which does not
	contain the~$k$-forbidden pattern, into an equivalent automaton which
	is~$k$-reversible.
	Furthermore, using a pumping argument, we prove that if an $n$-state
	automaton contains an~$N$-forbidden pattern, for a constant~$N=O(n^2)$,
	then it contains a $k$-forbidden pattern for each~$k>0$.
	Hence, applying this condition to the minimum automaton accepting a
	language~$L$, we are able to decide if~$L$ is weakly or strongly
	irreversible.
	We finally present a decision procedure for such problem.

	We point out that, according to the approach in~\cite{HJK15}, in this paper
	we refer to the \emph{classical model} of deterministic automata, namely
	automata with a unique initial state, a set of final states, and
	deterministic transitions.
	Different approaches have been considered in the literature.
	The notion of reversibility in~\cite{Ang82} is introduced by considering
	deterministic devices with one initial state and one final state, while
	automata with a set of initial states, a set of final states and
	deterministic transitions have been considered in~\cite{Pin92}.
	In particular, the notion of reversibility in~\cite{Ang82} is more
	restrictive than the one studied in~\cite{HJK15} and in this paper.
	Hence, also the notion of~$k$-reversibility, introduced and studied here,
	is different from a notion of~$k$-reversibility studied in~\cite{Ang82}.

	\section{Preliminaries}
	In this section we recall some basic definitions and results useful in the
	paper.
	For a detailed exposition, we refer the reader to~\cite{HU79}.

	Given a set~$S$, let us denote by~$\#S$ its cardinality, by~$2^S$ the
	family of all its subsets, and by~$S^{< k}$ ($S^k$, respectively), for a
	fixed integer~$k \geq 0$, the set of sequences of less than (exactly,
	resp.) $k$ elements from~$S$, where~$\varepsilon$ is the empty sequence.
	Given an alphabet~$\Sigma$, $|w|$~denotes the length of a string~$w \in
	\Sigma^*$.

	A \emph{deterministic finite automaton} (\dfa) is a tuple~$A=(Q,
	\Sigma, \delta, q_I, F)$, where~$Q$ is the finite set of \emph{states},
	$\Sigma$ is the \emph{input alphabet}, $q_I \in Q$ is the \emph{initial
	state}, $F \subseteq Q$ is the set of \emph{accepting} (or \emph{final})
	\emph{states}, and $\delta: Q \times \Sigma \rightarrow Q$ is the partial
	\emph{transition function}.
	A \emph{nondeterministic finite automaton} (\nfa) is an automaton in which
	it is possible to reach multiple states at the same time: multiple initial
	states are allowed and the transition function is defined as~$\delta: Q
	\times \Sigma \rightarrow 2^Q$.
	The \emph{language accepted} by an automaton is defined in classical way as
	the set of all strings that define a path from one initial state to one of
	the accepting states.

	Let~$A=(Q, \Sigma, \delta, q_I, F)$ be a \dfa.
	A state~$p\in Q$ is \emph{useful} if it is \emph{reachable}, i.e., there
	exists~$w \in \Sigma^*$ such that~$\delta(q_I, w) = p$, and
	\emph{productive}, i.e., if there is~$w \in \Sigma^*$ such
	that~$\delta(p,w) \in F$.
	In this paper we only consider automata with all useful states.

	The \emph{reverse} transition function of~$A$ is~$\delta^R: Q \times \Sigma
	\rightarrow 2^Q$, with $\delta^R(p, a) = \{q \in Q \mid \delta(q,a) = p\}$. 
	The \emph{reverse automaton}~$A^R = (Q,\Sigma, \delta^R, F, \{q_I\})$ is
	the \nfa\ obtained by reversing the transition function~$\delta$ and in
	which the set of initial states coincides with the set of final states
	of~$A$ and the unique final state is~$q_I$.

	A state~$r \in Q$ is said to be \emph{irreversible} when~$\#\delta^R(r,a) >
	1$ for some~$a\in\Sigma$, i.e., there are at least two transitions on the
	same letter entering~$r$, otherwise~$r$ is said to be \emph{reversible}.
	The \dfa\ $A$~is said to be \emph{irreversible} if it contains at least one
	irreversible state, otherwise $A$ is \emph{reversible} (\revdfa).
	As pointed out in~\cite{Kut15}, the notion of reversibility for a language
	is related to the computational model under consideration.
	In this paper we only consider \dfas.
	Hence, by saying that a language~$L$ is \emph{reversible}, we refer to this
	model, namely we mean that there exists a \revdfa\ accepting~$L$.
	The class of reversible languages is denoted by \rev.

	We say that two states~$p,q \in Q$ are \emph{equivalent} if and only if for
	all~$w\in \Sigma^*$, $\delta(p,w) \in F$ exactly when~$\delta(q,w)\in F$.
	Two automata~$A$ and~$A'$ are said to be \emph{equivalent} if they accept
	the same language, i.e., $L(A)=L(A')$.

	A \emph{strongly connected component} (\scc)~$C$ of a \nfa\ or a \dfa~$A$
	is a maximal subset of~$Q$ such that in the transition graph of~$A$ there
	exists a path between every pair of states in~$C$.
	Let us denote by~$\mathcal{C}_q$ the~\scc\ containing the state~$q \in Q$.

	We consider a partial order~$\preceq$ on the set of \sccs\ of~$A$, such
	that, for two such components~$C_1$ and~$C_2$, $C_1\preceq C_2$ when
	either~$C_1 = C_2$ or no state in~$C_1$ can be reached from a state
	in~$C_2$, but a state in~$C_2$ is reachable from a state in~$C_1$.
	We write~$C_1 \prec C_2$ when~$C_1\preceq C_2$ and~$C_1\neq C_2$.

	\section{Strong and weak irreversibility}
	In this section we introduce the main notions we consider in this paper, by
	defining strong and weak irreversibility and by presenting their basic
	properties.

	\begin{definition}
	\label{def:krev}
		Let~$k>0$ be an integer, $A=(Q,\Sigma,\delta,q_I,F)$ be a~\dfa,
		and~$L\subseteq\Sigma^*$ be a regular language.
		\begin{itemize}
			\item A state~$r \in Q$ is said to be \emph{$k$-irreversible} if
				there exist a string~$x\in\Sigma^{k-1}$ and a
				symbol~$\sigma\in\Sigma$, such that the cardinality  of the
				following set is greater than~$1$: \[\{\delta(p,x)\mid p\in
				Q\mbox{ and }\delta(p,x\sigma)=r\}\,.\]
				Otherwise, $r$ is said to be~\emph{$k$-reversible}.
			\item The automaton~$A$ is \emph{$k$-reversible} if each of its
				states is~$k$-reversible.
			\item The language~$L$ is~$k$-reversible if there exists
				a~$k$-reversible \dfa\ accepting it.
			\item The language~$L$ is \emph{weakly irreversible} if it
				is~$k$-reversible for some integer~$k > 0$.
			\item The language~$L$ is \emph{strongly irreversible} if it is not
				weakly irreversible.
		\end{itemize}
	\end{definition}
	By definition, a state~$r$ is~$1$-reversible if and only if it is
	reversible.
	As a consequence, $1$-reversibility coincides with reversibility.

	In the case of a~$k$-reversible state~$r$, with~$k>1$, we could have more
	than one transition on the same symbol~$\sigma$ entering~$r$.
	However, by knowing the suffix of length~$k$ of the part of the input
	already inspected, i.e., a suffix~$x\sigma$ with~$\lvert x \rvert = k-1$,
	we can uniquely identify which transition on~$\sigma$ has been used to
	enter~$r$ in the current computation.
	In other terms, while a reversible automaton is a device which is able to
	move the computation one state back, by knowing the last symbol that has
	been read, a~$k$-reversible automaton can do the same, having the suffix of
	length~$k$ of the part of the input already inspected (when the length of
	that part is less than~$k$, the automaton can see all the input so far
	inspected).

	\bigskip

	Let us denote by \krev{k} the class of~$k$-reversible languages.
	Hence, $\rev=\krev{1}$. Furthermore $k$-reversible \dfas\ are indicated as
	\krevdfas{k}, for short.

	From Definition~\ref{def:krev}, we can immediately prove the following
	facts:

	\begin{lemma}
		\label{lemma:krev}
		Let~$k>0$ be an integer, $A=(Q,\Sigma,\delta,q_I,F)$ be a~\dfa,
		and~$L\subseteq\Sigma^*$ be a regular language.
		\begin{itemize}
			\item If a state~$q\in Q$ is~$k$-reversible, then it
				is~$k'$-reversible for each~$k'>k$.
			\item If a state~$q\in Q$ is~$k$-irreversible, then it
				is~$k'$-irreversible for each~$k'<k$.
			\item If~$A$ is~$k$-reversible, then it is~$k'$-reversible for
				each~$k'>k$.
			\item If~$L$ is~$k$-reversible, then it is~$k'$-reversible for
				each~$k'>k$.
		\end{itemize}
	\end{lemma}

	\begin{example}{\textbf{\cite{KW14}}}
		\label{ex:prel}
		For each integer~$k > 0$, consider the language~$a^*b^kb^*$, which is
		accepted by the minimum automaton depicted in Figure~\ref{fig:inf-fam}.
		The only irreversible state is~$q_k$.

		Suppose that, after reading a string~$w$, the automaton is in~$q_k$.
		If we know a suffix of~$w$ of length~$i$, with~$i \leq k$, (this suffix
		can only be~$b^i$) then we cannot determine the previous state in the
		computation, i.e., the state before reading the last symbol of~$w$.
		In fact, this state could be either~$q_{k-1}$ or~$q_k$.
		Hence, the automaton is not~$k$-reversible.
		However, if we know the suffix of length~$k+1$, then it could be
		either~$b^{k+1}$, and in this case the previous state is~$q_k$,
		or~$ab^k$, and in this case the previous state is~$q_{k-1}$.
		It could be also possible that only~$k$ input symbols have been read,
		i.e., $\lvert w \rvert = k$.
		In that case, all~$w = b^k$ can be seen back and the previous state
		is~$q_{k-1}$.
		Hence, the automaton is~$(k+1)$-reversible.
		As shown in~\cite[Theorem~4]{KW14} we cannot do better for this
		language, i.e., $L \in \krev{k+1} \setminus \krev{k}$.
		This can be also obtained as a consequence of results in
		Section~\ref{sec:characterization}.
	\end{example}
	\begin{figure}[tb]
		\centering
		\begin{tikzpicture}[->,>=stealth',shorten >=1pt,auto,node distance=2cm,thick]
			\node[initial,initial text=,state] (q0) {$q_I$};
			\node[state] (q1) [right of=q0]	{$q_1$};
			\node[state] (q2) [right of=q1]	{$q_{k-1}$};
			\node[state,accepting] (q3) [right of=q2] {$q_k$};

			\path[every node/.style={font=\sffamily\small}]
			(q0) edge  [loop above] node {$a$} ()
			(q3) edge  [loop above] node {$b$} ()
			(q1) edge  [dashed] node {} (q2)
			(q2) edge  [] node {$b$} (q3)
			(q0) edge  [] node {$b$} (q1);
		\end{tikzpicture}
		\caption{The minimum automaton accepting the language~$a^*b^kb^*$}
		\label{fig:inf-fam}
	\end{figure}

	As a consequence of the last item in Lemma~\ref{lemma:krev} and of
	Example~\ref{ex:prel} we have the proper infinite hierarchy of classes
	\[\rev=\krev{1}\subset\krev{2}\subset\cdots\subset\krev{k}\subset\cdots\]

	In~\cite{HJK15}, the authors proved that a regular language is irreversible
	if and only if the minimum \dfa\ accepting it contains a  \emph{forbidden
	pattern}, which consists of two transitions of a same letter entering in a
	same state~$r$, where one of them arrives from a state~$p$ which belongs to
	the same strongly connected component of~$r$.
	We now refine such definition in order to consider strings of the same
	length that lead to the same state:

	\begin{definition}
		\label{def:kforb}
		Given a \dfa~$A=(Q,\Sigma,\delta,q_I,F)$ and an integer~$k>0$, the
		\emph{$k$-forbidden pattern} is formed by three states~$p,q,r\in Q$, with a
		symbol~$\sigma\in\Sigma$, two strings~$x\in\Sigma^{k-1}$
		and~$w\in\Sigma^*$, such that~$p\neq q$, $\delta(p,x)\neq\delta(q,x)$,
		$\delta(p,x\sigma)=\delta(q,x\sigma)=r$, and~$\delta(r,w)=q$.
	\end{definition}
	The~$k$-forbidden pattern just defined is depicted in
	Figure~\ref{fig:k-fp}.

	\begin{figure}[htbp]
		\centering
		\begin{tikzpicture}[->,>=stealth',shorten >=1pt,auto,thick,node
			distance=1.75cm,transform shape]
			\node[state] (r) {$r$};
			\node[state] (p') [above left=0.4cm and 1.75cm of r] {};
			\node[state] (q') [below left=0.4cm and 1.75cm of r] {};
			\node[scale=1.75] (d) at ($(p')!.5!(q')$) {$\neq$};
			\node[state] (p) [left=of p'] {$q$};
			\node[state] (q) [left=of q'] {$p$};

			\path[]
			(p) edge [] node {$x$} (p')
			(q) edge [] node {$x$} (q')
			(p') edge [pos=0.175] node {$\sigma$} (r)
			(q') edge [pos=0.5] node {$\sigma$} (r)
			(r) edge [
				loosely dashed,
				bend right=80,
				looseness=1.155,
				decorate,
				decoration={
					snake, 
					amplitude = 4mm,
					segment length = 15mm,
					post length=1mm
			}] node {$w$} (p);
		\end{tikzpicture}
		\caption{The $k$-forbidden pattern: $x\in\Sigma^{k-1}$, $\sigma\in\Sigma$, $w\in\Sigma^*$}
		\label{fig:k-fp}
	\end{figure}
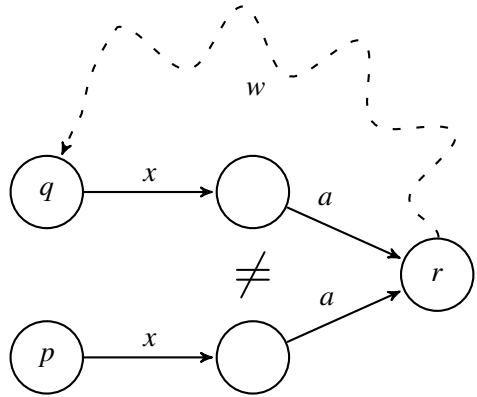

	\noindent 
	From Definition~\ref{def:kforb}, we can observe that if a \dfa~$A$ contains
	a~$k$-forbidden pattern, for some~$k>0$, then it contains a~$k'$-forbidden
	pattern for each~$k'$, with~$0 < k' < k$.

	The notion of~$k$-forbidden pattern will be used in the subsequent sections
	to obtain a characterization of the class \krev{k}. In fact, we will prove
	that a regular language is~$k$-reversible if and only if the minimum \dfa\
	accepting it does not contain the~$k$-forbidden pattern.

	\section{$k$-reversible simulation}
	\label{sec:simulation}

	In this section we present a construction to build, given a
	\dfa~$A=(Q,\Sigma,\delta,q_I,F)$ and an integer~$k>0$, an equivalent
	\dfa~$A'=(Q',\Sigma,\delta',q'_I,F')$, which is~$k$-reversible if~$A$ does
	not contain the~$k$-forbidden pattern.
	The \dfa~$A'$ simulates~$A$ by storing in its finite control three
	elements:
	\begin{itemize}
		\item The current state~$q$ of~$A$.
		\item An integer~$j\in\{1,\ldots,k\}$ which is used to count the
			first~$k$ visits to states in the current \scc\ of~$A$, namely in
			the \scc\ which contains the current state~$q$.
			When the value of the counter reaches~$k$, it is no more
			incremented, until a transition leaving the \scc.
			At that point, after saving its value in the third component of the
			state, $1$ is assigned to the counter for denoting the first visit
			in the \scc\ just reached.
		\item A sequence of pairs from~$Q\times\{1,\ldots,k\}$.
			This is the sequence of the first two components of the states
			in~$Q'$ which have been reached before simulating a transition that
			in~$A$ changes \scc.
			Since the number of possible \sccs\ is bounded by~$\#Q$, we
			consider sequences of length less than~$\#Q$.
	\end{itemize}
	Formally, we give the following definition:
	\begin{itemize}
		\item $Q' = Q \times \{1,\ldots,k\} \times (Q \times \{1, \ldots,
			k\})^{<\#Q}$,
		\item for~$\langle q,j,\alpha\rangle\in Q'$, if $\delta(q,a)=p$ then
			\[\delta'(\langle q,j,\alpha\rangle,a)=\left\{
					\begin{array}{ll}
						\langle p,\min\{j+1,k\}, \alpha \rangle &
							\mbox{if~${\cal C}_p={\cal C}_q$ }\\
						\langle p, 1, \alpha\cdot(q,j)\rangle &
							\mbox{otherwise,}\\
					\end{array}
				\right.
			\]
			while $\delta'(\langle q,j,\alpha\rangle,a)$ is not defined when
			$\delta(q,a)$ is not defined, and~$\cdot$ denotes the concatenation
			of a pair at the end of sequence,
		\item $q'_I=\langle q_I,1,\varepsilon\rangle$ is the initial state,
		\item $F' = F \times \{1, \ldots, k\} \times (Q \times \{1, \ldots,
			k\})^{<\#Q}$ is the set of final states.
	\end{itemize}
	Notice that by dropping the second and the third components off the states
	of~$A'$, we get exactly the automaton~$A$.
	Hence, $A$ and~$A'$ are equivalent.

	Furthermore, observe that if~$\delta'(\langle p,h,\alpha\rangle, a)=\langle
	r, \ell, \gamma \rangle$, with~$\langle p, h, \alpha \rangle, \langle r,
	\ell, \gamma \rangle \in Q'$, $a\in\Sigma$, and $1<\ell\leq k$, then the
	states~$p$ and~$r$ are in the same \scc\ of~$A$ and $h=\ell-1$ when
	$\ell<k$, while $h\in\{k-1,k\}$ when~$\ell=k$.
	This fact will be
	used in the following proof of the main property of~$A'$.

	\begin{lemma}
		\label{lemma:kfp}
		If~$A$ does not contain the~$k$-forbidden pattern, then~$A'$ is a $k$-reversible \dfa.
	\end{lemma}
	\begin{proof}
		By contradiction, let us suppose that~$A'$ contains a~$k$-irreversible state~$\langle r,\ell,\gamma\rangle\in Q'$.
		Then there exist a string~$x\in\Sigma^{k-1}$, a symbol~$a\in\Sigma$, states~$\langle p_0,h_0,\alpha_0\rangle,
		\langle p,h,\alpha\rangle, \langle q_0,j_0\,\beta_0\rangle, \langle q,j,\beta\rangle\in Q'$, such that
		$\delta'(\langle p_0,h_0,\alpha_0\rangle, x) = \langle p,h,\alpha\rangle$,
		$\delta'(\langle q_0,j_0,\beta_0\rangle, x) = \langle q,j,\beta\rangle$,
		$\delta'(\langle p,h,\alpha\rangle, a) = \delta'(\langle q,j,\beta\rangle, a) = \langle r,\ell,\gamma\rangle$,
		and~$\langle p,h,\alpha\rangle\neq\langle q,j,\beta\rangle$.
		The situation is summarized in the following picture:
		\[
			\begin{array}{cl}
				\langle p_0,h_0,\alpha_0\rangle\xrightarrow{~~x~~}\langle p,h,\alpha\rangle{\rotatebox[origin=l]{-30}{$\xrightarrow{~a~}$}}&~\\
				\phantom{\langle p_0,h_0,\alpha_0\rangle\xrightarrow{~~x~~}\langle p}\neq\phantom{\alpha\rangle\xrightarrow{~a~}}&\hspace*{-2ex}\langle r,\ell,\gamma\rangle\\
				\langle q_0,j_0,\beta_0\rangle\xrightarrow{~~x~~}\langle q,j,\beta\rangle{\rotatebox[origin=l]{30}{$\xrightarrow{~a~}$}}&~\\
			\end{array}
		\]
		For~$k>1$, the proof is divided in three cases, depending on the value
		of~$\ell$.
		\medskip

		\begin{itemize}

			\item \emph{Case~$\ell=1$.}\\
				Considering the definition of~$\delta'$, we notice that both states~$p$ and~$q$ are not in the same \scc\ of~$r$.
				Then $\gamma=\alpha\cdot(p,h)=\beta\cdot(q,j)$, thus implying $\alpha=\beta$, $p=q$, and~$h=j$. This is a contradiction to the hypothesis~$\langle p,h,\alpha\rangle\neq\langle q,j,\beta\rangle$. 

			\item \emph{Case~$1<\ell<k$.}\\
				Again from the definition of~$\delta'$, we can observe that
				$h=j=\ell-1<k-1$ and~$\alpha=\beta=\gamma$.
				We decompose~$x$ as $x'bx''$, where $x',x''\in\Sigma^*$,
				$b\in\Sigma$ and~$|x''|=\ell-2$.
				Then, in the paths on the string~$x$ from $\langle p_0, h_0,
				\alpha_0 \rangle$ to~$\langle p, \ell-1, \alpha \rangle$
				and from~$\langle q_0,j_0,\beta_0\rangle$ to~$\langle q,
				\ell-1, \alpha \rangle$ the last transitions that change \scc\ 
				in~$A$ are those on the symbol~$b$, immediately after the
				prefix~$x'$, i.e, we have the following situation:
				\[
					\begin{array}{cl}
						\langle p_0,h_0,\alpha_0\rangle\xrightarrow{~~x'~~}\langle p_1,h_1,\alpha_1\rangle\xrightarrow{~b~}\langle p_2,1,\alpha\rangle\xrightarrow{~~x''~~}\langle p,\ell-1,\alpha\rangle{\rotatebox[origin=l]{-30}{$\xrightarrow{~a~}$}}&~\\
																																																																																																					 &\hspace*{-2ex}\langle r,\ell,\alpha\rangle\\
						\langle q_0,j_0,\beta_0\rangle\xrightarrow{~~x'~~}\langle q_1,j_1,\beta_1\rangle\xrightarrow{~b~}\langle q_2,1,\alpha\rangle\xrightarrow{~~x''~~}\langle q,\ell-1,\alpha\rangle{\rotatebox[origin=l]{30}{$\xrightarrow{~a~}$}}&~\\
					\end{array}
				\]
				for suitable
				$\langle p_1,h_1,\alpha_1\rangle,\langle q_1,j_1,\beta_1\rangle\in Q'$, $p_2,q_2\in Q$.
				Then~$\alpha=\alpha_1\cdot(p_1,h_1)=\beta_1\cdot(q_1,j_1)$,
				that implies~$p_1=q_1$. As a consequence, since~$A$ is
				deterministic we get that $p_2=q_2$ and~$p=q$. Thus, also in
				this case we get the contradiction~$\langle
				p,h,\alpha\rangle=\langle q,j,\beta\rangle$.

			\item \emph{Case~$\ell=k$.}\\
				From the definition of~$\delta'$, we notice that either
				$h=j=k-1$, or at least one of~$h$ and~$j$ is equal to~$k$.
				In the first case, the proof can be completed as in the
				case~$1<\ell<k$, leading to a contradiction.
				In the case~$h=k$, moving backwards from the state $\langle
				p,k,\alpha\rangle$ to~$\langle p_0,h_0,\alpha_0\rangle$, along
				the transitions on the string~$x$ of length~$k-1$, we find a
				sequence of states whose all second components are equal
				to~$k$, which is followed by a (possibly empty) sequence of
				states where the values of the second components decrease
				by~$1$ at each transition.
				In this way we can conclude that~$h_0 \geq 1$ and all the first
				components, included~$p_0$, of states on this path, are in the
				same \scc\ of~$r$.
				Hence, $A$ contains the~$k$-forbidden pattern.
				The case~$j=k$ is similar.
		\end{itemize}
		For~$k=1$, if~$p$ or~$q$ are in the same \scc\ as~$r$ then~$A$ should
		contain the~$1$-forbidden pattern.
		Otherwise, we can proceed as in the case~$\ell = 1$, obtaining a
		contradiction.
	\end{proof}

	\noindent
	We now evaluate the size of the automaton obtained by using the previous
	construction.

	\begin{theorem}
		\label{th:states}
		Each~$n$-state \dfa\ which does not contain the~$k$-forbidden pattern
		can be simulated by an equivalent~$k$-reversible \dfa\ with no more
		than~$(k+1)^{n-1}$ states.
	\end{theorem}
	\begin{proof}
		Let~$A$ be a given~$n$-state \dfa\ not containing the~$k$-forbidden
		pattern.
		According to Lemma~\ref{lemma:kfp}, the automaton~$A'$ obtained
		from~$A$ with the above presented construction is~$k$-reversible.
		Now, we are going to estimate the number of reachable states in it.

		First of all, we notice that if~$\langle q,\ell,\alpha\rangle$ is a
		reachable state of~$A'$ and
		$\alpha=((p_1,j_1),(p_2,j_2),\ldots,(p_h,j_h))$, then ${\cal
		C}_{p_1}\prec{\cal C}_{p_2}\prec\cdots\prec{\cal C}_{p_h}\prec{\cal
		C}_{q}$. Hence, since the ordering among states appearing in~$\alpha$
		is given by the ordering of \sccs\ in~$A$, we could represent~$\alpha$
		as a set.

		This also allows to interpret the state~$\langle q,\ell,\alpha\rangle$ as the function $f:Q\rightarrow\{0,1,\ldots,k\}$, such that for~$r\in Q$:
		\[f(r)=\left\{
			\begin{array}{ll}
				\ell& \mbox{if~$r=q$,}\\
				j_i& \mbox{if~$r=p_i$, $1\leq i\leq h$,}\\
				0& \mbox{otherwise.}\\
			\end{array}
		\right. \]
		By counting the number of possible functions, we obtain a~$(k+1)^n$
		upper bound for the number of reachable states in~$A'$.

		Now, we show how to reduce this bound to the one claimed in the
		statement of the theorem.

		The above presented simulation can be slightly refined by observing
		that while simulating states in the \scc\ of the initial
		state~$q_I$, it is not necessary to keep the counter.
		Furthermore, in each state~$\langle q,\ell,\alpha\rangle$ of~$Q'$,
		with~$q\notin{\cal C}_{q_I}$, the first element of~$\alpha$, which
		should represent a state in~${\cal C}_{q_I}$, is stored without the
		counter.
		Hence, the state~$\langle q,\ell,\alpha\rangle$ can be seen as a
		state in~${\cal C}_{q_I}$ (the first element of~$\alpha$) with a
		function~$f:Q\setminus{\cal C}_{q_I}\rightarrow\{0,1,\ldots,k\}$
		(representing the current state with its counter and the other
		pairs in~$\alpha$).
		Since the counter associated with the current state is always
		positive, $f$ cannot be the null function.
		Hence, the number of possible functions is bounded by~$(k+1)^{n-s}-1$,
		where~$s=\#{\cal C}_{q_I}$.
		Considering also the states which are used in~$Q'$ to simulate the
		states in~${\cal C}_{q_I}$, this gives at most~$s + s((k+1)^{n-s}-1)$
		many reachable states.
		For~$k>0$ this amount is bounded by~$(k+1)^{n-1}$.
	\end{proof}

	We point out that for~$k=1$, Theorem~\ref{th:states} gives a~$2^{n-1}$ upper bound, which matches with the bound for the conversion of \dfas\ into equivalent \revdfas, claimed in~\cite{HJK15}.
	In the same paper, a lower bound very close to such an upper bound was presented.

	\section{A characterization of~$k$-reversible languages}
	\label{sec:characterization}

	In this section we present a characterization of $k$-reversible languages based on the notion of~$k$-forbidden pattern.
	This characterization will be obtained by combining Theorem~\ref{th:states} with the following result.

	\begin{lemma}
		\label{lemma:k-fp}
		Let~$L$ be a regular language and~$k$ be a positive integer. If the minimum \dfa\ accepting~$L$ contains the~$k$-forbidden pattern, then $L\notin\krev{k}$.
	\end{lemma}
	\begin{proof} 
		Let $M=(Q,\Sigma, \delta, q_I, F)$ be the minimum \dfa\ accepting $L$. By hypothesis there exist $p,q,r\in Q$, $\sigma\in\Sigma$, $x\in \Sigma^{k-1}$, $w\in \Sigma^*$ such that $p\neq q$, $\delta(p,x) \neq \delta(q,x)$, $\delta(p,x\sigma)= \delta(q,x\sigma)=r$ and $\delta(r,w)=q$.  Let $s = \delta(p,x)$ and $t = \delta(q,x)$.
		We are going to prove that each \dfa\ $A'=(Q',\Sigma, \delta', q'_I, F')$ accepting $L$ contains a $k$-irreversible state. 

		Let $q_0\in Q'$ be a state equivalent to $p$. In $A'$ we consider two arbitrarily long sequences of states $q_1, q_2, \ldots$ and $r_1, r_2, \ldots$ equivalent to $q$ and $r$, respectively,  such that $\delta'(q_{h-1},x\sigma)=r_h$ and $\delta'(r_h,w)=q_h$, for $h>0$. 
		Since $Q'$ is finite, sooner or later we will find an index $j$ such that either $r_i = r_j$ or $q_i = q_j$, for some $1\le i < j$. Let us take the first $j$ with such property. 
		\begin{itemize}

			\item Suppose $r_i = r_j$.
				If $i=1$, let $\hat{s}=\delta'(q_0,x)$ and $\hat{t}=\delta'(q_{j-1},x)$. Since $q_0$ is equivalent to $p$ and $q_{j-1}$ is equivalent to $q$, $\hat{s}$ and $\hat{t}$ are equivalent to the states $s$ and $t$ of $M$, respectively. So, $\hat{s}\neq \hat{t}$. Furthermore, $\delta'(\hat{s},\sigma)=\delta'(\hat{t},\sigma)=r_1$. Hence, $r_1$ is $k$-irreversible. In the case $i>1$, since $j$ is the first index giving a repetition we get $q_{i-1}\neq q_{j-1}$. We decompose the string $x\sigma$ as $x'\gamma x''$, where $x',x''\in \Sigma^*$, $\gamma\in\Sigma$, and $\delta'(q_{i-1},x') \neq \delta'(q_{j-1},x')$, $\delta'(q_{i-1},x'\gamma)=\delta'(q_{j-1},x'\gamma)=u$ for some $u\in Q'$ and $\delta'\!(u,x'')=r_i$. 
				We observe that $\delta'\!(q_{i-2}, x\sigma wx')=\delta'\!(q_{i-1},x') \neq \delta'\!(q_{j-2}, x\sigma wx')=\delta'\!(q_{j-1},x')$, while $\delta'\!(q_{i-2}, x\sigma wx'\!\gamma) =\delta'\!(q_{j-2}, x\sigma wx'\!\gamma)=u$.
				This implies that the state~$u$ is $|x\sigma wx'\!\gamma|$-irreversible.
				Hence it is $k$-irreversible.

			\item In the case $q_i = q_j$ and $r_i\neq r_j$, we observe that since $q_0$ is equivalent to $p$ and $q_j$ is equivalent to $q$ for $j\ge 1$, while $p$ and $q$ are not equivalent, we get $i>0$. We decompose $w$ as $w'\gamma w''$, where $w',w''\in\Sigma^*$, $\gamma\in\Sigma$ and $\delta'(r_i,w')\neq \delta'(r_j,w')$, $\delta'(r_i,w'\gamma)=\delta'(r_j,w'\gamma)=u$ for some $u\in Q'$ and $\delta'(u,w'')=q_i$.
				Then, $\delta'(q_{i-1},x\sigma w') \neq \delta'(q_{j-1},x\sigma w')$ and $\delta'(q_{i-1},x\sigma w'\gamma)=\delta'(q_{j-1},x\sigma w'\gamma)=u$. Hence, the state $u$ is $|x\sigma w' \gamma|$-irreversible, so it is $k$-irreversible.\qedhere
		\end{itemize}
	\end{proof}
\noindent
	Notice that the condition in Lemma~\ref{lemma:k-fp} is on the \emph{minimum \dfa} accepting the language under consideration. If we remove the requirement that the considered \dfa\ has to be minimum, the statement becomes false.
	For instance, the language $L=a^*$ is reversible even though for each~$k>0$ we can build a \dfa\ accepting it, which contains the~$k$-forbidden pattern (see Figure~\ref{fig:nonmin-k-fp}). 

	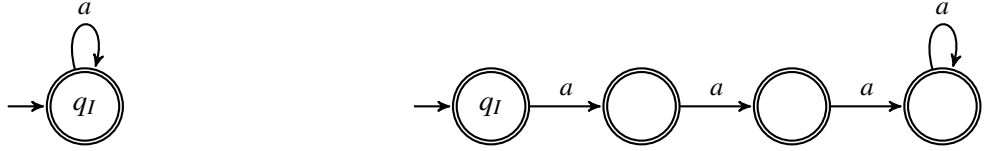
\begin{figure}[tb]
		\centering
		\begin{minipage}{.25\textwidth}
			\centering
			\begin{tikzpicture}[->,>=stealth',shorten >=1pt,auto,node
				distance=2cm, thick]
				\node[initial,initial text=,state,accepting] (q0) {$q_I$};
				\path[every node/.style={font=\sffamily\small}]
					(q0) edge  [loop above] node[anchor=south] {$a$} (q0);
			\end{tikzpicture}
		\end{minipage}\hfill
		\begin{minipage}{.7\textwidth}
			\centering
			\begin{tikzpicture}[->,>=stealth',shorten >=1pt,auto,node distance=2cm,
				thick]

				\node[initial,initial text=,state,accepting] (q0) {$q_I$};
				\node[state,accepting] (q1) [right of=q0]	{};
				\node[state,accepting] (q2) [right of=q1]	{};
				\node[state,accepting] (q3) [right of=q2]	{};

				\path[every node/.style={font=\sffamily\small}]
					(q0) edge  [] node[] {$a$} (q1)
					(q1) edge  [] node[] {$a$} (q2)
					(q2) edge  [] node[] {$a$} (q3)
					(q3) edge  [loop above] node[] {$a$} (q3);
			\end{tikzpicture}
		\end{minipage}
		\vskip1em
		\caption{
			The minimum \dfa\ accepting the reversible language $a^*$, and an
			equivalent \dfa\ containing the $3$-forbidden pattern
		}
		\label{fig:nonmin-k-fp}
	\end{figure}
	We are now able to characterize $k$-reversible languages in terms of the
	structure of minimum \dfas: 

	\begin{theorem}
		\label{th:kforb}
		Let~$L$ be a regular language. Given~$k>0$, $L\in\krev{k}$ if and only if the minimum \dfa\ accepting~$L$ does not contain the~$k$-forbidden pattern.
	\end{theorem}
	\begin{proof}
		The \emph{if} part is a consequence of Theorem~\ref{th:states},
		the \emph{only-if} part derives from Lemma~\ref{lemma:k-fp}.
	\end{proof}

	From Theorem~\ref{th:kforb}, we observe that to transform each \dfa~$A$
	accepting a~$k$-reversible language into an equivalent \krevdfa{k}, firstly
	we can transform~$A$ into the equivalent minimum \dfa~$M$ and then we can
	apply to~$M$ the construction presented in Section~\ref{sec:simulation}.

	As a consequence of Theorem~\ref{th:kforb} we also obtain:

	\begin{corollary}
		\label{cor:hforb}
		$L\in\krev{k+1}\setminus\krev{k}$ if and only if the maximum~$h$ such
		that the minimum \dfa\ accepting~$L$ contains the~$h$-forbidden pattern
		is~$k$.
	\end{corollary}

	In the following result we present further families of languages, besides
	that in Example~\ref{ex:prel}, which witness the existence of the proper
	infinite hierarchy
	\[\rev=\krev{1}\subset\krev{2}\subset\cdots\subset\krev{k}\subset\cdots\]
	Furthermore, we show that the difference between the ``amount'' of
	irreversibility in a minimum \dfa\ and in the accepted language can be
	arbitrarily large:

	\begin{theorem}
		For all integers~$k,j>0$ with~$j>k>1$ there exists a language~$L_{k,j}$
		such that:
		\begin{itemize}
			\item The minimum \dfa\ accepting~$L_{k,j}$ is a~\krevdfa{j} but
				not a~\krevdfa{j-1}.
			\item $L_{k,j}\in\krev{k}\setminus\krev{k-1}$.
		\end{itemize}
	\end{theorem}
	\begin{proof}
		Let~$L_{k,j}$ be the language accepted by the automaton~$A_{k,j} = (Q,
		\Sigma, \delta, q_I, F)$ where~$\Sigma = \{a,b\}$,
		$Q = \{q_I, q'_1, q''_1, \ldots, q'_{j-1}, q''_{j-1}, q_j\}$, $F=
		\{q''_{j-1},q_j\}$, and the transition function is defined as follows
		(see Figure~\ref{fig:L5,7} for an example):
		\begin{itemize}
			\item $\delta(q_I, a) = q'_1$
			\item $\delta(q_I, b) = q''_1$
			\item $\delta(q'_i,a)=q'_{i+1}$ and~$\delta(q''_i,a)= q''_{i+1}$
				for~$1\leq i \leq j-k$
			\item $\delta(q'_i,b)=q'_{i+1}$ and~$\delta(q''_i,b)= q''_{i+1}$
				for~$j-k < i < j-1$
			\item $\delta(q'_{j-1},b) = \delta(q''_{j-1},b) =
				\delta(q''_{j-1},a) = \delta(q_j,b)= q_j$
		\end{itemize}

		Firstly, we can observe that~$A_{k,j}$ is the minimum \dfa\ 
		accepting~$L_{k,j}$.
		It contains only one irreversible state, $q_j$, with~$\delta^R(q_j,b)
		= \{q_j,q'_{j-1},q''_{j-1}\}$.
		We also notice that~$\delta(q'_1, a^{j-k}b^{k-2}) = q'_{j-1} \neq
		q''_{j-1} = \delta(q''_1, a^{j-k}b^{k-2})$, while~$\delta(q'_1,
		a^{j-k}b^{k-1}) = \delta(q''_1, a^{j-k}b^{k-1}) = q_j$.
		Hence~$A_{k,j}$ is not a \krevdfa{j-1}.
		However, the knowledge of one more symbol in the suffix of the input
		read to enter~$q_j$ allows to determine the state of the automaton
		before reading the last symbol.
		In particular, if the suffix of length~$j$ is~$a^{j-k+1}b^{k-1}$, then
		the state was~$q'_{j-1}$; if the suffix is~$ba^{j-k}b^{k-1}$
		or~$ba^{j-k}b^{k-2}a$, then the state was~$q''_{j-1}$; in the remaining
		cases it was~$q_j$.
		Hence, $A_{k,j}$ is a \krevdfa{j}.

		To prove that~$L_{k,j}\in\krev{k}\setminus\krev{k-1}$, we first show
		that~$A_{k,j}$ contains the~$(k-1)$-forbidden pattern.
		To this aim, in Definition~\ref{def:kforb} we can choose~$q=r=q_j$,
		$p=q'_{j-k+1}$, $\sigma=b$, $x=b^{k-2}$ and~$w = \varepsilon$.
		Furthermore, it is possible to obtain a \krevdfa{k}~$A'_{k,j}$
		equivalent to~$A_{k,j}$ by duplicating~$q_j$ with its loop and by
		redistributing incoming transitions from~$q'_{j-1}$ and~$q''_{j-1}$, as
		in the case presented in Figure~\ref{fig:A'5,7}.
	\end{proof}

	\begin{figure}[tb]
		\centering
		\begin{tikzpicture}[->, >=stealth', shorten >=1pt, auto, thick, node
			distance=1.75cm]
			\node[initial,initial text=,state] (q0) {$q_I$};
			\node[state] (q1) [above right of=q0] {$q'_1$};
			\node[state] (q2) [below right of=q0] {$q''_1$};
			\node[state] (q3) [right of=q1]	{$q'_2$};
			\node[state] (q4) [right of=q2]	{$q''_2$};
			\node[state] (q5) [right of=q3]	{$q'_3$};
			\node[state] (q6) [right of=q4]	{$q''_3$};
			\node[state] (q7) [right of=q5]	{$q'_4$};
			\node[state] (q8) [right of=q6]	{$q''_4$};
			\node[state] (q9) [right of=q7]	{$q'_5$};
			\node[state] (q10) [right of=q8]	{$q''_5$};
			\node[state] (q11) [right of=q9]	{$q'_6$};
			\node[state,accepting] (q12) [right of=q10]	{$q''_6$};
			\node[state,accepting] (q13) [below right of=q11]	{$q_7$};

			\path[every node/.style={font=\sffamily\small}]
				(q13) edge  [loop below] node {$b$} ()
				(q12) edge  [] node {$a,b$} (q13)
				(q10) edge  [] node {$b$} (q12)
				(q8) edge  [] node {$b$} (q10)
				(q6) edge  [] node {$b$} (q8)
				(q4) edge  [] node {$a$} (q6)
				(q2) edge  [] node {$a$} (q4)
				(q11) edge  [] node {$b$} (q13)
				(q9) edge  [] node {$b$} (q11)
				(q7) edge  [] node {$b$} (q9)
				(q5) edge  [] node {$b$} (q7)
				(q3) edge  [] node {$a$} (q5)
				(q1) edge  [] node {$a$} (q3)
				(q0) edge  [] node {$b$} (q2)
				(q0) edge  [] node {$a$} (q1);
		\end{tikzpicture}
		\caption{The minimum automaton~$A_{5,7}$ accepting the language~$L_{5,7}$}
		\vspace{0.75em}
		\label{fig:L5,7}
	\end{figure}
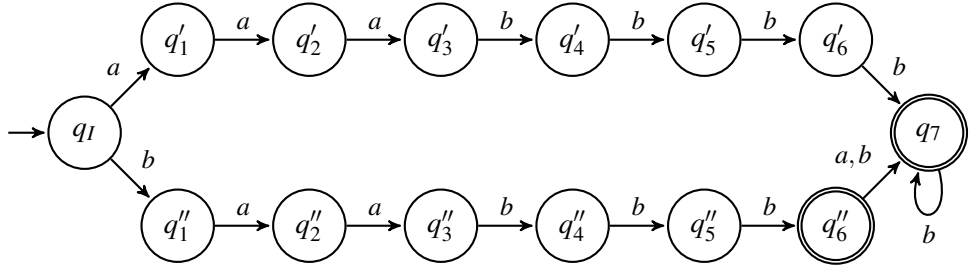
	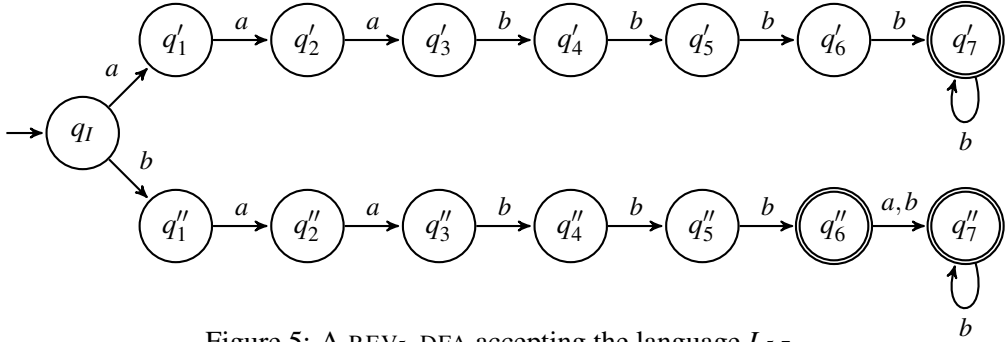
\begin{figure}[tb]
		\centering
		\begin{tikzpicture}[->,>=stealth',shorten >=1pt,auto,node distance=1.75cm,thick]
			\node[initial,initial text=,state] (q0) {$q_I$};
			\node[state] (q1) [above right of=q0] {$q'_1$};
			\node[state] (q2) [below right of=q0] {$q''_1$};
			\node[state] (q3) [right of=q1]	{$q'_2$};
			\node[state] (q4) [right of=q2]	{$q''_2$};
			\node[state] (q5) [right of=q3]	{$q'_3$};
			\node[state] (q6) [right of=q4]	{$q''_3$};
			\node[state] (q7) [right of=q5]	{$q'_4$};
			\node[state] (q8) [right of=q6]	{$q''_4$};
			\node[state] (q9) [right of=q7]	{$q'_5$};
			\node[state] (q10) [right of=q8]	{$q''_5$};
			\node[state] (q11) [right of=q9]	{$q'_6$};
			\node[state,accepting] (q12) [right of=q10]	{$q''_6$};
			\node[state,accepting] (q13) [right of=q11]	{$q'_7$};
			\node[state,accepting] (q14) [right of=q12]	{$q''_7$};

			\path[every node/.style={font=\sffamily\small}]
			(q14) edge  [loop below] node {$b$} ()
			(q13) edge  [loop below] node {$b$} ()
			(q12) edge  [] node {$a,b$} (q14)
			(q10) edge  [] node {$b$} (q12)
			(q8) edge  [] node {$b$} (q10)
			(q6) edge  [] node {$b$} (q8)
			(q4) edge  [] node {$a$} (q6)
			(q2) edge  [] node {$a$} (q4)
			(q11) edge  [] node {$b$} (q13)
			(q9) edge  [] node {$b$} (q11)
			(q7) edge  [] node {$b$} (q9)
			(q5) edge  [] node {$b$} (q7)
			(q3) edge  [] node {$a$} (q5)
			(q1) edge  [] node {$a$} (q3)
			(q0) edge  [] node {$b$} (q2)
			(q0) edge  [] node {$a$} (q1);
		\end{tikzpicture}
		\vspace{-2em}
		\caption{A \krevdfa{5} accepting the language~$L_{5,7}$}
		\label{fig:A'5,7}
	\end{figure}

	\section{Weakly and strongly irreversible languages}
	\label{sec:strongly}
	By Definition~\ref{def:krev}, a language is \emph{weakly irreversible} if
	it is~$k$-reversible for some~$k>0$, namely if it is in the
	class~$\bigcup_{k>0}\krev{k}$.
	A natural question is whether or not the class of weakly irreversible
	languages coincides with the class of regular languages.
	In this section we will give a negative answer to this question, thus
	proving the existence of strongly irreversible languages.

	First of all, we observe that, by Theorem~\ref{th:kforb}, a regular language is \emph{strongly irreversible} if and only if the minimum \dfa\ accepting it contains a~$k$-forbidden pattern for each~$k>0$. Using a combinatorial argument, we now prove that in order to decide if a language is strongly or weakly irreversible, it is enough to consider only a value of~$k$ which depends on the size of the minimum \dfa:

	\begin{theorem}
		\label{th:longforb}
		Let~$A=(Q,\Sigma, \delta, q_I,F)$ be an $n$-state \dfa\ and $N>\frac{n^2-n}{2}$. If~$A$ contains an $N$-forbidden pattern, then it contains a $k$-forbidden pattern for each $k> 0$.
	\end{theorem}
	\begin{proof}
		Suppose that~$A$ contains an~$N$-forbidden pattern.
		As observed after Definition~\ref{def:kforb}, $A$ contains a $k$-forbidden pattern for each $k\le N$.

		We now prove the same for $k>N$. By hypothesis there exist $p,q,r\in Q$, $\sigma\in\Sigma$, $x\in \Sigma^{N-1}$, $w\in \Sigma^*$, such that $p\neq q$, $\delta(p,x) \neq \delta(q,x)$, $\delta(p,x\sigma)= \delta(q,x\sigma)=r$, and $\delta(r,w)=q$. Let $x= \sigma_1\sigma_2\cdots\sigma_{N-1}$ with $\sigma_i\in \Sigma$, for $i=1,\ldots, N-1$.  
		Moreover, let $p_0,\ldots, p_{N-1}, q_0,\ldots, q_{N-1}\in Q$ be such that $p=p_0$, $q=q_0$, $p_i = \delta(p_{i-1}, \sigma_i)$, $q_i = \delta(q_{i-1}, \sigma_i)$ for $i=1,\ldots,N-1$, and $\delta(p_{N-1},\sigma)=\delta(q_{N-1},\sigma)=r$. 
		Since~$p_{N-1}\neq q_{N-1}$ and $A$ is deterministic, we get~$p_i\neq q_i$ for $i=0,\ldots, N-1$.  
		Notice that there are $n^2-n$ possible pairs of different states. 

		We consider the pairs $(p_0,q_0), \ldots, (p_{N-1},q_{N-1})$.
		Since~$N>(n^2-n)/2$ and $p_i\neq q_i$, for $i=0,\ldots, N-1$, there are two  indices $i,j$, $0\le i< j \le N-1$ such that either $(p_i,q_i)=(p_j,q_j)$ or $(p_i,q_i)=(q_j,p_j)$. 
		So $\delta(p_i, (\sigma_{i+1}\cdots\sigma_{j})^2)=p_i$ and $\delta(q_i, (\sigma_{i+1}\cdots\sigma_{j})^2)=q_i$. 
		Given $h>0$, we consider the string $z_h = \sigma_1\cdots\sigma_i (\sigma_{i+1}\cdots\sigma_j)^{2h} \sigma_{j+1}\cdots\sigma_{N-1}$. We can verify that $\delta(p,z_h)=p_{N-1}$ and $\delta(q,z_h)= q_{N-1}$. This implies that $A$ contains the $|z_h|+1$-forbidden pattern.
		Since~$i\neq j$, by properly choosing $h$, this allows us to obtain a $k$-forbidden pattern for each arbitrarily large $k$. 
	\end{proof}

	Combining Theorem~\ref{th:kforb} with Theorem~\ref{th:longforb} we obtain:

	\begin{corollary}
		\label{cor:strongly-irrev}
		Let~$L$ be a regular language whose minimum~\dfa\ has~$n$ states.
		Then~$L$ is strongly irreversible if and only if it is
		not~$(\frac{n^2-n}{2}+1)$-reversible.
	\end{corollary}

	We now present an example of strongly irreversible language.

	\begin{example}
		The language~$L=a^*b(a+b)^*$ is strongly irreversible.
		The minimum automaton accepting it has~$2$ states (see
		Figure~\ref{fig:strongly-irrev}).
		We notice that~$\delta(q_I, ab) = \delta(p, ab) = p$,
		while~$\delta(q_I, a) \neq \delta(p,a)$.
		This defines a~$2$-forbidden pattern.
		According to Corollary~\ref{cor:strongly-irrev}, this implies that~$L$
		is strongly irreversible.
		Observe that entering in~$p$ with each string~$a^kb$, we have
		a~$(k+1)$-forbidden pattern, for any~$k \geq 0$.
	\end{example}

	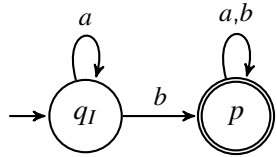
\begin{figure}[tb]
		\centering
		\begin{tikzpicture}[->,>=stealth',shorten >=1pt,auto,node
			distance=2cm,thick]
			\node[initial,initial text=,state] (q0) {$q_I$};
			\node[state,accepting] (q1) [right of=q0]	{$p$};

			\path[every node/.style={font=\sffamily\small}]
			(q0) edge  [loop above] node {$a$} ()
			(q1) edge  [loop above] node {$a$,$b$} (q1)
			(q0) edge  [] node {$b$} (q1);
		\end{tikzpicture}
		\caption{The minimum automaton accepting the language~$L=a^*b(a+b)^*$}
		\label{fig:strongly-irrev}
	\end{figure}

	\section{Decision problems}
	In this section we provide a method to decide whether a language~$L$ is
strongly or weakly irreversible, and, in the latter case, to find the
minimum~$k$ such that~$L$ is $k$-reversible.

The idea is to simultaneously analyze all the paths entering each irreversible
state~$r \in Q$ of the minimum automaton~$A$ accepting~$L$ in order to find the
longest string~$z$ that, with at least two different paths, leads to~$r$ and
defines the~$\lvert z \rvert$-forbidden pattern or to discover that there exist
arbitrarily long strings with such property.
This corresponds to analyze all couples of paths starting from two different
states~$p,q\in Q$ that, with the same string~$z$, lead to~$r$.
Intuitively, this can be done by constructing the product automaton of two
copies of the reversal automaton of~$A$, i.e., $A^R \times A^R$, and by
analyzing all paths starting from the states of the form~$(r,r)$.
Since the goal is to establish the nature of the (ir)reversability of~$L$ ---
not of~$A$ --- it is useful to recall that by Definition~\ref{def:kforb} it is
enough to consider only the couples of paths in which one of them is completely
included in the same \scc\ of~$r$, i.e., $\mathcal{C}_r=\mathcal{C}_q$.
To this aim, we are going to consider the product between~$A^R$ and a
transformation of~$A^R$ which is obtained by splitting it in \sccs.

Let~$A=(Q,\Sigma,\delta,q_I,F)$ be an irreversible \dfa, $A^R = (Q,\Sigma,
\delta^R, F, \{q_I\})$ be the reversal automaton of~$A$, and~$A^R_{\sccs} =
(Q,\Sigma,\delta_{\sccs}^R,F,\{q_I\})$ be the \nfa\ obtained by splitting~$A^R$
in its \sccs, i.e., $\delta_{\sccs}^R(r,a) = \{ q \mid q\in\delta^R(r,a) \text{
and } \mathcal{C}_r = \mathcal{C}_q\}$, for~$r \in Q$, $a\in\Sigma$.
Let us define the automaton~$\hat A = A^R \times A^R_{\sccs}$ as follows:
$\hat A = (\hat Q,\Sigma,\hat\delta,\hat I,\hat F)$ where $\hat Q = \hat F = Q
\times Q$, $\hat I = \{(r,r) \mid r \in Q\}$, and $\hat \delta((r',r''),a) =
\{(p,q) \in \delta^R(r',a) \times \delta^R_{\sccs}(r'',a) \mid p \neq q\}$.

The resulting automaton~$\hat A$ accepts all strings~$z$ which
define a~$\lvert z \rvert$-forbidden pattern (plus the empty string).
Formally, this follows from the following lemma, whose proof can be given by
induction:
\begin{lemma}
	Consider a path~$(r,r),(p_1,q_1),\ldots,(p_{\lvert z \rvert-1},q_{\lvert z
	\rvert-1}), (p,q)$ in~$\hat A$ from a state~$(r,r)$ to~$(p,q)$ on a
	string~$z$.
	Then~$\hat\delta((r,r),z) \ni (p,q)$ if and only if all the following
	conditions are satisfied:
	\begin{enumerate}
		\item $p_i \neq q_i$ for each~$0 < i < \lvert z \rvert$, $p \neq q$,
		\item $\delta(p,z) = r$,
		\item $\delta(q,z) = r$ and $\mathcal{C}_r = \mathcal{C}_q$.
	\end{enumerate}
\end{lemma}
Considering Theorem~\ref{th:kforb}, this leads to state the following 
\begin{lemma}
	Let $A$ be a minimum~$n$-state \dfa\ and~$\hat A$ be the \nfa\ defined as
	above.
	Then:
	\begin{itemize}
		\item The following statements are equivalent:
			\begin{itemize}
				\item $A$ is strongly irreversible,
				\item $L(\hat A)$ is an infinite language,
				\item $L(\hat A)$ contains a string of length~$\frac{n^2-n}{2} + 1$.
			\end{itemize}
		\item For each~$k>0$, $L(A) \in \krev{k}$ if and only if~$L(\hat A)$
			contains only strings of length less than~$k$.
	\end{itemize} 
\end{lemma}
The same argument can be exploited to prove that the problem of checking
whether~$L(A)$ is strongly or weakly irreversible is in~$\textsc{NL}$, namely
the class of problems accepted by nondeterministic logarithmic space bounded
Turing machines.

\begin{theorem}
	The problem of deciding whether a language is strongly or weakly
	irreversible is~$\textsc{NL}$-complete.
\end{theorem}
\begin{proof}(sketch)
	Given a minimum \dfa\ accepting the language under consideration and the
	above described automaton~$\hat A$, the problem can be reduced to testing
	if the transition graph of~$\hat A$ contains at least one loop.
	In such a case, there are arbitrarily long strings in~$L(\hat A)$, namely
	strings describing $k$-forbidden patterns for arbitrarily large~$k$,
	and~$L(A)$ is strongly irreversible.
	The problem of verifying the existence of a loop is in~$\textsc{NL}$.

	To prove the~$\textsc{NL}$-completeness, we show a reduction from the
	\emph{Graph Accessibility Problem} ($\emph{GAP}$) which
	is~$\textsc{NL}$-complete (for further details see~\cite{Jones75}).
	Let~$G=(V,E)$ be a directed graph where~$V = \{1, \ldots, n\}$.
	Our goal is to define a \dfa~$A$\ such that $A$ is strongly irreversible if
	and only if there exists a path from~$1$ to~$n$ in~$G$.
	We build~$A'$ by starting from the same ``state structure'' of~$G$, and
	adding a \scc\ providing the forbidden pattern when combined with a path
	from~$1$ to~$n$ in the original graph.

	We stress that the instance of our problem should be an automaton
	containing only useful states, while automata that can be ``intuitively''
	obtained from GAP instances could have useless states and, detecting them,
	would require to solve GAP.

	Let $A' = (Q, \Sigma, \delta, q_I, \{ q_F \})$ be a \dfa\ where~$Q = V \cup
	\{ q_I, q_F, q_1,\ldots,q_{n-1}\}$, $\Sigma =\{0, \ldots, n, \$, \sharp\}$,
	and~$\delta$ is defined as follows:
	\begin{enumerate}[label=\roman*.]
		\item\label{gap:original}
			$\delta(i,j) = j$ for~$(i,j) \in E$, $i \neq j$
		\item\label{gap:clique}
			$\delta(q_i,j) = q_j$ for~$0 < i,j < n$, $i \neq j$
		\item\label{gap:n}
			$\delta(q_i,n) = n$ for~$0 < i < n$
		\item\label{gap:scc}
			$\delta(n,0) = q_1$
		\item\label{gap:useful}
			$\delta(q_I,i) = i$ and~$\delta(i,\sharp) = q_F$ for~$0 < i \leq n$
		\item\label{gap:loop}
			$\delta(1,\$) = 1$ and~$\delta(q_1, \$) = q_1$.
	\end{enumerate}
	Observe that the restriction of the underlying graph~$A'$ to
	states~$1,\ldots,n$ coincides with~$G$ (transitions \ref{gap:original}).
	In addition, the set of states~$\{q_1, \ldots, q_{n-1}\}$ extends the
	\scc~$\mathcal{C}_n$ so that each state can reach the others
	in~$\mathcal{C}_n$ with a single transition (transitions~\ref{gap:clique},
	\ref{gap:n}, and~\ref{gap:scc}).
	This implies that the state~$n$ is reachable from~$q_1$ with all the
	possible paths passing through the states in the \scc.
	Furthermore, a loop is added to states~$1$ and~$q_1$ on the symbol~$\$$ in
	order to create a forbidden pattern (transitions~\ref{gap:loop}).
	Notice that each state in~$Q$ is useful (transitions~\ref{gap:useful}).

	In such a way, the states~$\{1,n,q_1\}$ form a forbidden pattern with
	strings of arbitrary length if and only if the given graph contains a path
	from~$n$ to~$1$.
	Notice that any state~$i\in Q\setminus\{n\}$ is, at most, $1$-irreversible.
	So we can conclude that~$A'$ is strongly irreversible if and only if there
	exists a path from~$1$ to~$n$ in~$G$.

	It can be shown that the reduction can be computed in deterministic
	logarithmic space.
\end{proof}
 
	\section{Conclusion}

	We introduced and studied the notions of strong and weak irreversibility for finite automata and regular languages. 
	In Section~\ref{sec:characterization} we proved the existence of an infinite hierarchy of weakly irreversible languages, while in Section~\ref{sec:strongly} we showed the existence of strongly irreversible languages, namely of regular languages that are not weakly irreversible. In both cases, the witness languages are defined over a binary alphabet, so the question arises if the same results hold in the case of a one-letter alphabet, i.e., in the case of unary languages. We now briefly discuss this point.

	First of all, we remind the reader that the transition graph of a unary \dfa\ consists of an initial path, which is followed by a loop (for a recent survey on unary automata, we address the reader to~\cite{Pig15}). Hence, a unary \dfa\ is reversible if and only if the initial path is of length~$0$, i.e., the automaton consists only of a loop (in this case the accepted language is said to be \emph{cyclic}). We can also observe that given an integer~$k>0$, a unary language is~$k$-reversible if and only if it is accepted by a \dfa\ with an initial path of less than~$k$ states. Hence, for each~$k$, the language~$a^{k-1}a^*$ is~$k$-reversible, but not~$(k-1)$-reversible. This shows the existence of an infinite hierarchy of weakly irreversible languages even in the unary case.
	Furthermore, from the above discussion, we can observe that if a unary language is accepted by a~\dfa\ with an initial path of~$k$ states, then it is $(k+1)$-reversible. This implies that each unary regular language is weakly irreversible (see also~\cite[Proposition~10]{KW14}).
	Hence, to obtain strongly irreversible languages, we need alphabets of at least two letters.

	\medskip
	The definition of~$k$-reversible automata and languages have been given for each integer~$k>0$. One could ask if it does make sense to consider a notion of~$0$-reversibility. According to the interpretation we gave to~$k$-reversibility, a state is~$0$-reversible when in each computation its predecessor can be obtained by knowing the last~$0$ symbols which have been read from the input, i.e., without the knowledge of any previous input symbol. This means that a~$0$-irreversible state can have only one entering transition, or no entering transitions if it is the initial state. As a consequence, the transition graph of a~$0$-reversible automaton is a tree rooted in the initial state and~$0$-reversible languages are exactly finite languages.

	\section*{Acknowledgment}
	We thank the anonymous referees for valuable suggestions, in particular for
	addressing us to consider the results obtained in~\cite{KW14}.

	\bibliographystyle{eptcs}
	\bibliography{biblio}  
\end{document}